\documentclass{article}
\usepackage{spconf,amsmath,amssymb,amsthm}
\usepackage[dvips]{graphicx}
\usepackage{algorithm}
\usepackage{algorithmicx}
\usepackage{algpseudocode}

\newcommand{\A}{\mathrm{A}}
\newcommand{\B}{\mathrm{B}}
\newcommand{\ato}{\overset{\mathrm{a.s.}}{\to}}

\newtheorem{assumption}{Assumption}
\newtheorem{theorem}{Theorem}

\title{Long-Memory Message-Passing for Spatially Coupled Systems}
\name{Keigo Takeuchi\thanks{The author was in part supported by the Grant-in-Aid 
for Scientific Research~(B) (JSPS KAKENHI Grant Numbers 21H01326), Japan.}}
\address{Toyohashi University of Technology, 
Dept.\ Electrical and Electronic Inf.\ Eng., 
Toyohashi, Japan}
\begin{document}
%
\maketitle
\begin{abstract}
This paper addresses the reconstruction of sparse signals from spatially 
coupled, linear, and noisy measurements. A unified framework of rigorous 
state evolution is established for developing long-memory message-passing 
(LM-MP) in spatially coupled systems. LM-MP utilizes all previous messages to 
compute the current message while conventional MP only uses the latest 
messages. The unified framework is utilized to propose orthogonal approximate 
message-passing (OAMP) for spatially coupled systems. The framework for 
LM-MP is used as a technical tool to prove the convergence of state 
evolution for OAMP. Numerical results show that OAMP for spatially coupled 
systems is superior to that for systems without spatial 
coupling in the so-called waterfall region.
\end{abstract}
\section{Introduction}
This paper addresses the reconstruction of $L$ unknown $N$-dimensional 
sparse signal vectors $\{\boldsymbol{x}[l]\in\mathbb{R}^{N}: 
l\in\mathcal{L}_{0}=\{0,\ldots,L-1\}\}$ from spatially 
coupled $M$-dimensional measurements 
$\{\boldsymbol{y}[\ell]\in\mathbb{R}^{M}: \ell\in\mathcal{L}_{W}
=\{0,\ldots,L+W-1\}\}$ with coupling 
width~$W$~\cite{Kudekar11,Krzakala12,Donoho13,Takeuchi15}, given 
by\footnote{Any variable with indices is regarded as zero outside the domain 
of the indices. For instance, $\boldsymbol{A}[\ell][l]=\boldsymbol{O}$ holds 
for all $(\ell, l)\notin\mathcal{L}_{W}\times\mathcal{L}_{0}$. 
} 
\begin{equation} \label{spatial_coupling}
\boldsymbol{y}[\ell] 
= \sum_{w=0}^{W}\gamma[\ell][\ell-w]\boldsymbol{A}[\ell][\ell-w]
\boldsymbol{x}[\ell-w] + \boldsymbol{w}[\ell].
\end{equation}
In (\ref{spatial_coupling}), $\boldsymbol{A}[\ell][l]\in\mathbb{R}^{M\times N}$ 
denotes a known sensing matrix in row section $\ell\in\mathcal{L}_{W}$ and 
column section~$l\in\mathcal{L}$. The additive white Gaussian noise (AWGN) 
vector $\boldsymbol{w}[\ell]\sim\mathcal{N}(\boldsymbol{0},
\sigma^{2}\boldsymbol{I}_{M})$ has independent zero-mean Gaussian elements 
with variance $\sigma^{2}$. The coupling coefficient 
$\gamma[\ell][l]\in\mathbb{R}$ satisfies the normalization condition 
$L^{-1}\sum_{l=0}^{L-1}\sum_{w=0}^{W}\gamma^{2}[l+w][l] = 1$. 
The spatially coupled system~(\ref{spatial_coupling}) may be regarded as a 
mathematical model of cell-free massive multiple-input 
multiple-output (MIMO)~\cite{Ngo17} or sparse superposition 
codes~\cite{Joseph12,Barbier17,Rush21}. 

Approximate message-passing (AMP)~\cite{Krzakala12,Donoho13,Takeuchi15} is 
low-complexity message-passing (MP) for signal recovery in spatially 
coupled systems. When the sensing matrices $\{\boldsymbol{A}[\ell][l]\}$ 
have independent and identically distributed (i.i.d.) elements with zero mean, 
Bayes-optimal AMP was proved to achieve the information-theoretic compression 
limit in the noiseless case~\cite{Donoho13,Bayati11,Javanmard13}. However, 
AMP fails to converge for the other sensing matrices, such as the 
ill-conditioned~\cite{Rangan191} or non-zero mean~\cite{Vila15} case. 

Orthogonal AMP~\cite{Ma17} or equivalently vector AMP~\cite{Rangan192} solves 
this convergence issue in AMP: Bayes-optimal orthogonal/vector AMP was proved 
to converge~\cite{Takeuchi221,Liu221} and achieve the Bayes-optimal 
performance~\cite{Rangan192,Takeuchi201} for orthogonally invariant 
sensing matrices---a generalization of zero-mean i.i.d.\ Gaussian matrices. 
However, Bayes-optimal orthogonal/vector AMP requires the 
high-complexity linear minimum mean-square error (LMMSE) filter.  

Long-memory (LM) MP~\cite{Opper16,Takeuchi202,Takeuchi21,Liu222,Skuratovs22,Fan22,Venkataramanan22} is an attractive approach to realize 
the advantages of both AMP and orthogonal/vector AMP: low complexity and 
Bayes-optimality. LM-MP in \cite{Liu222,Skuratovs22} utilizes the same 
low-complexity filter as AMP and all previous messages to approximate the 
LMMSE filter while orthogonal/vector AMP uses the LMMSE 
filter to construct a sufficient statistic only from the latest 
messages~\cite{Takeuchi221,Liu221}. As a result, LM-MP in \cite{Liu222} was 
proved to achieve the Bayes-optimal performance for orthogonally invariant 
sensing matrices. Since the complexity issue is outside the main scope of 
this paper, however, it is left as future research for spatially coupled 
systems.  
 
Another advantage of LM-MP is the convergence guarantee: Use of LM 
damping~\cite{Liu222} guarantees the convergence of state evolution for 
LM-MP~\cite{Takeuchi221,Liu221}. This property of LM-MP was utilized as a 
technical tool to prove the convergence of state evolution for Bayes-optimal 
orthogonal/vector AMP without memory~\cite{Takeuchi221,Liu221}. This paper 
uses LM-MP as a tool for the convergence guarantee.  

The purpose of this paper is to establish a unified framework of state 
evolution for LM-MP in spatially coupled systems. The framework is a 
generalization of state evolution for spatially coupled 
i.i.d.\ Gaussian matrices~\cite{Javanmard13} to the orthogonal invariance 
and LM-MP cases. It is also a generalization of state evolution 
for the orthogonal invariance and LM-MP cases~\cite{Takeuchi21} to the 
spatial coupling case. 

As a memoryless instance of the unified framework, this paper proposes 
orthogonal AMP (OAMP) for spatially coupled systems. The framework for LM-MP 
is utilized to prove the convergence of state evolution for the proposed OAMP.  
Numerical results are presented to show the superiority of OAMP for 
spatially coupled systems to that for conventional systems without spatial 
coupling. 

\section{Unified Framework of State Evolution}
To present a unified framework of state evolution for the spatially coupled 
system~(\ref{spatial_coupling}), we transform (\ref{spatial_coupling}) into 
a vector system. Let $\mathcal{W}[\ell]=\{\max\{\ell-(L-1),0\},\ldots,
\min\{W,\ell\}\}$ denote the set of indices for row section~$\ell$ such that  
$\{\boldsymbol{A}[\ell][\ell-w]: w\in\mathcal{W}[\ell]\}$ are non-zero 
matrices in row section~$\ell$ of the spatially coupled 
system~(\ref{spatial_coupling}). For $N_{\mathrm{c}}[\ell]=|\mathcal{W}[\ell]|N$, 
the matrix $\boldsymbol{A}[\ell]\in\mathbb{R}^{M\times N_{\mathrm{c}}[\ell]}$ consists 
of the normalized matrices 
$|\mathcal{W}[\ell]|^{-1/2}\boldsymbol{A}[\ell][\ell-w]$ for 
all $w\in\mathcal{W}[\ell]$. Then, the spatially coupled 
system~(\ref{spatial_coupling}) is transformed into 
\begin{equation} \label{vector_system}
\boldsymbol{y}[\ell] 
= \boldsymbol{A}[\ell]\vec{\boldsymbol{x}}[\ell] + \boldsymbol{w}[\ell], 
\end{equation}
with
\begin{equation} \label{x_vec}
\vec{\boldsymbol{x}}[\ell] 
= \sqrt{|\mathcal{W}[\ell]}
\mathrm{vec}\{
\gamma[\ell][\ell-w]\boldsymbol{x}[\ell-w]: w\in\mathcal{W}[\ell]\}, 
\end{equation}
where the notation $\mathrm{vec}\{\boldsymbol{v}_{i}: i=1,\ldots,n\}$ denotes 
the column $(\boldsymbol{v}_{1}^{\mathrm{T}},\ldots,
\boldsymbol{v}_{n}^{\mathrm{T}})^{\mathrm{T}}$. The system~(\ref{vector_system}) 
looks like parallel systems without spatial coupling when we ignore 
the dependencies between $\{\vec{\boldsymbol{x}}[\ell]\}$ through the 
signal vectors $\{\boldsymbol{x}[l]\}$ in (\ref{x_vec}). 

We present the notation used in the proposed framework of state evolution. 
Consider the singular-value decomposition (SVD) 
$\boldsymbol{A}[\ell]=\boldsymbol{U}[\ell]\boldsymbol{\Sigma}[\ell]
\boldsymbol{V}^{\mathrm{T}}[\ell]$. We define 
$\vec{\boldsymbol{\lambda}}[\ell]\in\mathbb{R}^{N_{\mathrm{c}}[\ell]}$ as 
the vector that consists of all diagonal elements of 
$\boldsymbol{\Sigma}^{\mathrm{T}}[\ell]\boldsymbol{\Sigma}[\ell]$. 
For $N_{\mathrm{c}}=(W+1)N$, we define the set of extended vectors $\Lambda=
\{[\vec{\boldsymbol{\lambda}}^{\mathrm{T}}[\ell],
\boldsymbol{0}^{\mathrm{T}}]^{\mathrm{T}}\in\mathbb{R}^{N_{\mathrm{c}}}: 
\ell\in\mathcal{L}_{W}\}$. The extended vectors are used for a unified 
treatment of $\{\vec{\boldsymbol{\lambda}}[\ell]\}$ in both bulk and boundary 
sections. Similarly, we define the set of extended noise vectors as 
$\Omega=\{[(\boldsymbol{U}^{\mathrm{T}}[\ell]\boldsymbol{w}[\ell])^{\mathrm{T}},
\boldsymbol{0}^{\mathrm{T}}]^{\mathrm{T}}\in\mathbb{R}^{N_{\mathrm{c}}}: 
\ell\in\mathcal{L}_{W}\}$. The set of extended signal vectors is written 
as $\mathcal{X}=\{[(\boldsymbol{S}_{\mathrm{x}}[\ell]\mathrm{vec}
\{\boldsymbol{x}[l]: l\in\mathcal{L}_{0}\})^{\mathrm{T}}, 
\boldsymbol{0}^{\mathrm{T}}]^{\mathrm{T}}\in\mathbb{R}^{N_{\mathrm{c}}}: 
\ell\in\mathcal{L}_{W}\}$ for deterministic selection matrices 
$\{\boldsymbol{S}_{\mathrm{x}}[\ell]\in\{0, 1\}^{N_{\mathrm{c}}[\ell]\times LN}\}$, 
which select $N_{\mathrm{c}}[\ell]$ different elements from an $LN$-dimensional 
vector multiplied from the right side. 
The selection matrices are used for flexibility of the unified framework.  

Suppose that the dynamics of estimation errors for LM-MP can be described 
with a dynamical system with respect to four vectors 
$\{\vec{\boldsymbol{b}}_{t}[\ell], 
\vec{\boldsymbol{m}}_{t}^{\mathrm{ext}}[\ell], \vec{\boldsymbol{h}}_{t}[\ell], 
\vec{\boldsymbol{q}}_{t}^{\mathrm{ext}}[\ell]\in\mathbb{R}^{N_{\mathrm{c}}[\ell]}\}$ 
for iteration~$t$ and row section~$\ell\in\mathcal{L}_{W}$. We define the 
matrix $\vec{\boldsymbol{B}}_{t}[\ell]=(\vec{\boldsymbol{b}}_{0}[\ell],\ldots,
\vec{\boldsymbol{b}}_{t-1}[\ell])\in\mathbb{R}^{N_{\mathrm{c}}[\ell]\times t}$ 
for all previous iterations and the set of extended matrices 
$\mathcal{B}_{t}=\{[\vec{\boldsymbol{B}}_{t}^{\mathrm{T}}[\ell],
\boldsymbol{O}]^{\mathrm{T}}\in\mathbb{R}^{N_{\mathrm{c}}\times t}: 
\ell\in\mathcal{L}_{W}\}$. Similarly, we define 
$\vec{\boldsymbol{H}}_{t}[\ell]\in\mathbb{R}^{N_{\mathrm{c}}[\ell]\times t}$ and 
$\mathcal{H}_{t}$. For two vector-valued functions 
$\boldsymbol{\phi}_{t}[\ell]:\mathbb{R}^{N_{\mathrm{c}}\times (t+3)|\mathcal{L}_{W}|}
\to\mathbb{R}^{N_{\mathrm{c}}}$ and $\boldsymbol{\psi}_{t}[\ell]:
\mathbb{R}^{N_{\mathrm{c}}\times (t+2)|\mathcal{L}_{W}|}\to\mathbb{R}^{N_{\mathrm{c}}}$, 
the proposed dynamical system with a general initial condition 
$\vec{\boldsymbol{q}}_{0}^{\mathrm{ext}}[\ell]
=(\boldsymbol{I}_{N_{\mathrm{c}}[\ell]}, \boldsymbol{O})
\boldsymbol{\psi}_{-1}[\ell](\mathcal{X})$ is given by 
\begin{align}
\vec{\boldsymbol{b}}_{t}[\ell]
=&\boldsymbol{V}^{\mathrm{T}}[\ell]\vec{\boldsymbol{q}}_{t}^{\mathrm{ext}}[\ell],
\label{b} \\
\vec{\boldsymbol{m}}_{t}^{\mathrm{post}}[\ell]
=& (\boldsymbol{I}_{N_{\mathrm{c}}[\ell]}, \boldsymbol{O})
\boldsymbol{\phi}_{t}[\ell](\mathcal{B}_{t+1}, \Omega, \Lambda),  
\label{m_post} \\
\vec{\boldsymbol{m}}_{t}^{\mathrm{ext}}[\ell] 
=& \vec{\boldsymbol{m}}_{t}^{\mathrm{post}}[\ell]
- \sum_{\tau=0}^{t}\xi_{\A,\tau,t}[\ell]\vec{\boldsymbol{b}}_{\tau}[\ell], 
\label{m_ext} \\
\vec{\boldsymbol{h}}_{t}[\ell] 
=& \boldsymbol{V}[\ell]\vec{\boldsymbol{m}}_{t}^{\mathrm{ext}}[\ell],
\label{h} \\
\vec{\boldsymbol{q}}_{t+1}^{\mathrm{post}}[\ell]
=& (\boldsymbol{I}_{N_{\mathrm{c}}[\ell]}, \boldsymbol{O})
\boldsymbol{\psi}_{t}[\ell](\mathcal{H}_{t+1}, \mathcal{X}),  
\label{q_post}
 \\
\vec{\boldsymbol{q}}_{t+1}^{\mathrm{ext}}[\ell] 
=& \vec{\boldsymbol{q}}_{t+1}^{\mathrm{post}}[\ell]
- \sum_{\tau=0}^{t}\xi_{\B,\tau,t}[\ell]\vec{\boldsymbol{h}}_{\tau}[\ell], 
\label{q_ext}
\end{align}
The coefficients $\xi_{\A,\tau,t}[\ell]$ and $\xi_{\B,\tau,t}[\ell]$ 
in the Onsager correction~(\ref{m_ext}) and (\ref{q_ext}) 
are defined as   
\begin{align}
\xi_{\A,\tau,t}[\ell] 
=& \left\langle
\frac{\partial}{\partial \vec{\boldsymbol{b}}_{\tau}[\ell]}
\boldsymbol{\phi}_{t}[\ell](\mathcal{B}_{t+1}, \Omega, \Lambda)
\right\rangle, \\
\xi_{\B,\tau,t}[\ell] 
=& \left\langle
\frac{\partial}
{\partial\vec{\boldsymbol{h}}_{\tau}[\ell]}
\boldsymbol{\psi}_{t}[\ell](\mathcal{H}_{t+1}, \mathcal{X})
\right\rangle,
\label{xi_B}
\end{align}
where for $\boldsymbol{f}:\mathbb{R}^{n}\to
\mathbb{R}^{n}$ we have used  
$[\partial\boldsymbol{f}(\boldsymbol{x})/\partial\boldsymbol{x}]_{i}
=\partial [\boldsymbol{f}]_{i}/\partial [\boldsymbol{x}]_{i}$  and 
$\langle \partial\boldsymbol{f}(\boldsymbol{x})/\partial\boldsymbol{x}\rangle
=n^{-1}\sum_{i=1}^{n}\partial [\boldsymbol{f}]_{i}/\partial [\boldsymbol{x}]_{i}$. 

The general error model~(\ref{b})--(\ref{q_ext}) for the spatial coupling 
case is a generalization of a conventional error model without spatial 
coupling~\cite{Takeuchi21}. The conventional error model was utilized to 
propose existing LM-MP 
algorithms~\cite{Takeuchi202,Takeuchi21,Liu222,Skuratovs22}. 
By designing the two functions $\phi_{t}[\ell]$ and $\psi_{t}[\ell]$ 
appropriately, the general error model~(\ref{b})--(\ref{q_ext}) can represent 
the dynamics of estimation errors for LM-MP.   

To present state evolution for the general error 
model~(\ref{b})--(\ref{q_ext}), we postulate the following assumptions: 
\begin{assumption} \label{assumption_x} 
For some $\epsilon>0$, the signal vector $\boldsymbol{x}[l]$ has i.i.d.\ 
elements with zero mean, unit variance, and a bounded $(2+\epsilon)$th 
moment.  
\end{assumption}
\begin{assumption} \label{assumption_A}
The sensing matrices $\{\boldsymbol{A}[\ell]\}$ in (\ref{vector_system}) are 
independent for all $\ell$. Each $\boldsymbol{A}[\ell]$ is right-orthogonally 
invariant: $\boldsymbol{V}[\ell]$ in the SVD is independent 
of $\boldsymbol{U}[\ell]\boldsymbol{\Sigma}[\ell]$ and Haar-distributed. 
Furthermore, the empirical eigenvalue distribution of $|\mathcal{W}[\ell]|
\boldsymbol{A}^{\mathrm{T}}[\ell]\boldsymbol{A}[\ell]$ converges almost surely 
to a compactly supported deterministic distribution with unit mean 
in the large system limit, where both $M$ and $N$ tends to infinity with 
the compression rate $\delta=M/N\in(0, 1]$ kept constant.
\end{assumption}
\begin{assumption} \label{assumption_function}
The function $\boldsymbol{\phi}_{t}[\ell]$ is separable with respect to 
all variables and proper\footnote{
Let $L_{i}>0$ denote a Lipschitz-constant for a $k$th-order pseudo-Lipschitz 
function $f_{i}:\mathbb{R}^{n}\to\mathbb{R}$: For all 
$\boldsymbol{x}, \boldsymbol{y}\in\mathbb{R}^{n}$, $|f_{n}(\boldsymbol{x})
-f_{n}(\boldsymbol{y})|\leq L_{n}(1 + \|\boldsymbol{x}\|^{k-1} 
+ \|\boldsymbol{y}\|^{k-1})\|\boldsymbol{x} - \boldsymbol{y}\|$ holds. 
The function $\boldsymbol{f}=(f_{1},\ldots,f_{n})^{\mathrm{T}}$ is said to be 
proper if $\limsup_{n\to\infty}n^{-1}\sum_{i=1}^{n}L_{i}^{j}<\infty$ holds for all 
$j\in\mathbb{N}$.} Lipschitz-continuous with 
respect to $\mathcal{B}_{t+1}$ and $\Omega$ while $\boldsymbol{\psi}_{t}[\ell]$ 
is separable and proper Lipschitz-continuous with respect to all variables.  
Furthermore, $\|\vec{\boldsymbol{m}}_{t}^{\mathrm{ext}}[\ell]\|\neq0$ and 
$\|\vec{\boldsymbol{q}}_{t+1}^{\mathrm{ext}}[\ell]\|\neq0$ hold for all $t$. 
\end{assumption}
\begin{theorem} \label{theorem_SE} 
Postulate Assumptions~\ref{assumption_x}, \ref{assumption_A}, 
and~\ref{assumption_function}, and suppose that  
$\tilde{\boldsymbol{\phi}}_{t}(\mathcal{B}_{t+1},\Omega, \Lambda):
\mathbb{R}^{N_{\mathrm{c}}\times(t+3)|\mathcal{L}_{W}|}\to\mathbb{R}^{N_{\mathrm{c}}}$ 
is separable, second-order pseudo-Lipschitz with respect to 
$\mathcal{B}_{t+1}$ and $\Omega$, and proper. Similarly, suppose that 
$\tilde{\boldsymbol{\psi}}_{t}(\mathcal{H}_{t+1}, 
\mathcal{X}):\mathbb{R}^{N_{\mathrm{c}}\times(t+2)|\mathcal{L}_{W}|}
\to\mathbb{R}^{N_{\mathrm{c}}}$ 
is a separable, second-order pseudo-Lipschitz, and proper function. Then,  
\begin{align} 
\langle\tilde{\boldsymbol{\phi}}_{t}(\mathcal{B}_{t+1}, \Omega, \Lambda)
\rangle
-& \mathbb{E}\left[
 \langle\tilde{\boldsymbol{\phi}}_{t}(\mathcal{Z}_{\A,t+1}, 
 \Omega,\Lambda)\rangle 
\right]\ato 0, 
\label{phi_SLLN} \\
\langle\tilde{\boldsymbol{\psi}}_{t}(\mathcal{H}_{t+1}, \mathcal{X})
\rangle
-& \mathbb{E}\left[
 \langle\tilde{\boldsymbol{\psi}}_{t}(\mathcal{Z}_{\B,t+1},\mathcal{X})
 \rangle 
\right]\ato 0. \label{psi_SLLN}
\end{align}
In (\ref{phi_SLLN}), the set $\mathcal{Z}_{\A,t+1}=\{
[\vec{\boldsymbol{Z}}_{\A,t+1}^{\mathrm{T}}[\ell], \boldsymbol{O}]^{\mathrm{T}}
\in\mathbb{R}^{N_{\mathrm{c}}\times(t+1)}: \ell\in\mathcal{L}_{W}\}$ is composed 
of independent matrices for all $\ell$. 
Each matrix $\vec{\boldsymbol{Z}}_{\A,t+1}[\ell]
=(\vec{\boldsymbol{z}}_{\A,0}[\ell],\ldots,
\vec{\boldsymbol{z}}_{\A,t}[\ell])\in\mathbb{R}^{N_{\mathrm{c}}[\ell]\times (t+1)}$ 
has zero-mean Gaussian random 
vectors with covariance $\mathbb{E}[\vec{\boldsymbol{z}}_{\A,\tau}[\ell]
\vec{\boldsymbol{z}}_{\A,\tau'}^{\mathrm{T}}[\ell]]
=c_{\tau',\tau}[\ell]\boldsymbol{I}_{N_{\mathrm{c}}[\ell]}$ for all 
$\tau', \tau\in\{0,\ldots,t\}$, with 
$N_{\mathrm{c}}^{-1}[\ell](\vec{\boldsymbol{q}}_{\tau'}^{\mathrm{ext}}[\ell])^{\mathrm{T}}
\vec{\boldsymbol{q}}_{\tau}^{\mathrm{ext}}[\ell]\ato c_{\tau',\tau}[\ell]$.
In (\ref{psi_SLLN}), $\mathcal{Z}_{\B,t+1}$ is 
defined in the same manner as for $\mathcal{Z}_{\A,t+1}$ with the 
exception of $N_{\mathrm{c}}^{-1}[\ell](\vec{\boldsymbol{m}}_{\tau'}^{\mathrm{ext}}
[\ell])^{\mathrm{T}}\vec{\boldsymbol{m}}_{\tau}^{\mathrm{ext}}[\ell]
\ato c_{\tau',\tau}[\ell]$. 
\end{theorem}
\begin{proof}
See \cite[Theorem 7]{Takeuchi222}. 
\end{proof}

Theorem~\ref{theorem_SE} implies asymptotic Gaussianity for 
$\vec{\boldsymbol{B}}_{t}[\ell]$ and $\vec{\boldsymbol{H}}_{t}[\ell]$, which 
is an important property in deriving state evolution recursions for LM-MP. 
By defining the two functions $\boldsymbol{\phi}_{t}[\ell]$ and 
$\boldsymbol{\psi}_{t}[\ell]$ in (\ref{m_post}) and (\ref{q_post}) 
appropriately, Theorem~\ref{theorem_SE} allows 
us to derive state evolution recursions for LM-MP. 

\begin{figure}[t]
\begin{algorithm}[H]
\caption{Orthogonal AMP with $T$ iterations} 
\label{alg1}
\begin{algorithmic}[1]
\State \label{initialization}
Let $\vec{\boldsymbol{x}}_{\B\to\A,0}[\ell]=\boldsymbol{0}$ and 
$v_{\B\to\A,0}[\ell]=\sum_{w\in\mathcal{W}[\ell]}\gamma^{2}[\ell][\ell-w]$ for all 
$\ell\in\mathcal{L}_{W}$.    
\For{$t=0,\ldots,T-1$}
\For{$\ell=0,\ldots,L+W-1$} \label{moduleA_start}

\State \label{LMMSE}
$\boldsymbol{W}_{t}[\ell]=(\sigma^{2}v_{\B\to\A,t}^{-1}[\ell]\boldsymbol{I}_{M} 
+ \boldsymbol{A}[\ell]\boldsymbol{A}^{\mathrm{T}}[\ell])^{-1}
\boldsymbol{A}[\ell]$.

\State \label{mean_A}
\hspace{-2em}$\vec{\boldsymbol{x}}_{\A,t}[\ell] 
= \vec{\boldsymbol{x}}_{\B\to\A,t}[\ell] + \boldsymbol{W}_{t}^{\mathrm{T}}[\ell]
(\boldsymbol{y}[\ell] - \boldsymbol{A}[\ell]
\vec{\boldsymbol{x}}_{\B\to\A,t}[\ell])$.

\State \label{eta_A}
$\eta_{\A,t}[\ell]= N_{\mathrm{c}}^{-1}[\ell]\mathrm{Tr}
(\boldsymbol{I}_{N_{\mathrm{c}}[\ell]} 
- \boldsymbol{W}_{t}^{\mathrm{T}}[\ell]\boldsymbol{A}[\ell])$.

\State \label{mean_AB}
$\vec{\boldsymbol{x}}_{\A\to\B,t}[\ell]
= \frac{\vec{\boldsymbol{x}}_{\A,t}[\ell] 
- \eta_{\A,t}[\ell]\vec{\boldsymbol{x}}_{\B\to\A,t}[\ell]}
{\sqrt{|\mathcal{W}[\ell]|}(1 - \eta_{\A,t}[\ell])}$. 

\State \label{var_AB}
$v_{\A\to\B,t}[\ell] 
= \frac{\eta_{\A,t}[\ell]v_{\B\to\A,t}[\ell]}
{|\mathcal{W}[\ell]|(1-\eta_{\A,t}[\ell])}$. 
\EndFor \label{moduleA_end}

\For{$l=0,\ldots,L-1$} \label{moduleB_start}
\State \label{var_suf}
$v_{\A\to \B,t}^{\mathrm{suf}}[l] 
= \left(
 \sum_{w=0}^{W}\frac{\gamma^{2}[l+w][l]}{v_{\A\to \B,t}[l+w]}
\right)^{-1}$.

\State \label{mean_suf}
Let $\vec{\boldsymbol{x}}_{\A\to\B,t}[l+w][w]\in\mathbb{R}^{N}$ denote 
the $w$th section in $\vec{\boldsymbol{x}}_{\A\to\B,t}[l+w]$ 
for $w\in\mathcal{W}[l+w]$ 
and compute $\boldsymbol{x}_{\A\to \B,t}^{\mathrm{suf}}[l]
= v_{\A\to \B,t}^{\mathrm{suf}}[l]\sum_{w=0}^{W}\gamma[l+w][l]
\frac{\vec{\boldsymbol{x}}_{\A\to\B,t}[l+w][w]}{v_{\A\to \B,t}[l+w]}$.  

\State \label{mean_B}
$\boldsymbol{x}_{\B,t+1}[l]
=f_{\mathrm{opt}}(\boldsymbol{x}_{\A\to \B,t}^{\mathrm{suf}}[l];
v_{\A\to \B,t}^{\mathrm{suf}}[l])$. 

\State \label{var_B}
$v_{\B,t+1}[l]
=\langle \mathrm{Var}(\boldsymbol{x}_{\A\to \B,t}^{\mathrm{suf}}[l];
v_{\A\to \B,t}^{\mathrm{suf}}[l]) \rangle$.
\EndFor

\For{$\ell=0,\ldots,L+W-1$}
\State \label{mean_B_extend}
$\vec{\boldsymbol{x}}_{\B,t+1}[\ell]
=\sqrt{|\mathcal{W}[\ell]}\mathrm{vec}\{
\gamma[\ell][\ell-w]\boldsymbol{x}_{\B,t+1}[\ell-w]: w\in\mathcal{W}[\ell]\}$. 

\State \label{eta_B}
$\eta_{\B,t}[\ell]
= \sum_{w\in\mathcal{W}[\ell]}
\frac{\gamma^{2}[\ell][\ell-w]
v_{\B,t+1}[\ell-w]}{v_{\A\to\B,t}[\ell]}$.  

\State \label{mean_BA}
$\vec{\boldsymbol{x}}_{\B\to\A,t+1}[\ell] = \frac{\sqrt{|\mathcal{W}[\ell]|}
\vec{\boldsymbol{x}}_{\B,t+1}[\ell] - \eta_{\B,t}[\ell]
\vec{\boldsymbol{x}}_{\A\to\B,t}[\ell]}
{\sqrt{|\mathcal{W}[\ell]|}(1 - \eta_{\B,t}[\ell]/|\mathcal{W}[\ell]|)}$. 

\State \label{var_BA}
$v_{\B\to\A,t+1}[\ell] 
= \frac{\eta_{\B,t}[\ell]v_{\A\to\B,t}[\ell]}
{1 - \eta_{\B,t}[\ell]/|\mathcal{W}[\ell]|}$.  

\State \label{mean_damp}
\hspace{-1em}$\vec{\boldsymbol{x}}_{\B\to\A,t+1}[\ell]:=
\zeta\vec{\boldsymbol{x}}_{\B\to\A,t+1}[\ell]
+(1-\zeta)\vec{\boldsymbol{x}}_{\B\to\A,t}[\ell]$. 

\State \label{var_damp}
$v_{\B\to\A,t+1}[\ell]:= \zeta v_{\B\to\A,t+1}[\ell] 
+ (1-\zeta)v_{\B\to\A,t}[\ell]$. 

\EndFor \label{moduleB_end}
\EndFor
\State Output $\boldsymbol{x}_{\mathrm{B},T}[l]$ as an estimator of 
$\boldsymbol{x}[l]$ for all $l\in\mathcal{L}_{0}$. 
\end{algorithmic}
\end{algorithm}
\end{figure}

\section{Orthogonal AMP}
We propose Bayes-optimal OAMP for the spatially coupled 
system~(\ref{vector_system}) on the basis of Theorem~\ref{theorem_SE}. 
See Algorithm~\ref{alg1} for the details of the proposed algorithm. 
Lines~\ref{moduleA_start}--\ref{moduleA_end} and 
lines~\ref{moduleB_start}--\ref{moduleB_end} correspond to 
the LMMSE estimation---called module~A---and element-wise nonlinear 
estimation---called module~B---respectively. 

Module~A computes the LMMSE estimator of $\vec{\boldsymbol{x}}[\ell]$ in 
line~\ref{mean_A}. Then, the Onsager correction in lines~\ref{mean_AB} 
and~\ref{var_AB} is performed to realize the asymptotic Gaussianity for the 
error vector 
$\vec{\boldsymbol{h}}_{t}[\ell]=\vec{\boldsymbol{x}}_{\A\to\B,t}[\ell] 
- |\mathcal{W}[\ell]|^{-1/2}\vec{\boldsymbol{x}}[\ell]$. 

Module~B transforms the message $\vec{\boldsymbol{x}}_{\A\to\B,t}[\ell]$ 
in the extended signal space $\mathbb{R}^{N_{\mathrm{c}}[\ell]}$ into 
the sufficient statistic $\boldsymbol{x}_{\A\to \B,t}^{\mathrm{suf}}[l]$ for  
the original signal vector $\boldsymbol{x}[l]\in\mathbb{R}^{N}$ in 
line~\ref{mean_suf}. After denoising in line~\ref{mean_B} and transforming 
$\boldsymbol{x}_{\B,t+1}[l]\in\mathbb{R}^{N}$ into 
$\vec{\boldsymbol{x}}_{\B,t+1}[\ell]$ for the extended signal space 
in line~\ref{mean_B_extend}, the Onsager correction is computed in 
lines~\ref{mean_BA} and \ref{var_BA}. Lines~\ref{mean_damp} and 
\ref{var_damp} represent the damping steps with $\zeta\in(0,1]$ to improve 
the convergence property of OAMP for finite $M$ and $N$. 

The functions $f_{\mathrm{opt}}$ and $\mathrm{Var}$ in lines~\ref{mean_B} and 
\ref{var_B} denote the Bayes-optimal denoiser 
$f_{\mathrm{opt}}(u;v)=\mathbb{E}[x_{1}[0] | x_{1}[0]+\sqrt{v}z=u]$ and 
the corresponding posterior variance $\mathrm{Var}(u;v)
=\mathbb{E}[\{x_{1}[0]-f_{\mathrm{opt}}(u;v)\}^{2}| x_{1}[0]+\sqrt{v}z=u]$, 
with $z\sim\mathcal{N}(0,1)$. The validity of these definitions is 
justified in the large system limit via Theorem~\ref{theorem_SE}. We have 
used a popular notation in the MP community for lines~\ref{mean_B} and 
\ref{var_B}: For a scalar function $f: \mathbb{R}\to\mathbb{R}$ the notation 
$f(\boldsymbol{x})$ denotes the element-wise application of $f$ to the vector 
$\boldsymbol{x}$, i.e.\ $[f(\boldsymbol{x})]_{i}=f([\boldsymbol{x}]_{i})$.  

Assume no damping $\zeta=1$ in lines~\ref{mean_damp} and \ref{var_damp} 
for theoretical analysis. As proved in \cite[Lemma~8]{Takeuchi222}, 
the dynamics of the error vectors  
$\vec{\boldsymbol{h}}_{t}[\ell]=\vec{\boldsymbol{x}}_{\A\to\B,t}[\ell] 
- |\mathcal{W}[\ell]|^{-1/2}\vec{\boldsymbol{x}}[\ell]$ and 
$\vec{\boldsymbol{q}}_{t}^{\mathrm{ext}}[\ell]=
\vec{\boldsymbol{x}}_{\B\to\A,t}[\ell] - \vec{\boldsymbol{x}}[\ell]$ can be 
described with the general error model~(\ref{b})--(\ref{q_ext}). State 
evolution results for OAMP are obtained via the unified framework of 
state evolution in Theorem~\ref{theorem_SE}.  

We first define state evolution recursions for OAMP. To distinguish 
variables in the state evolution recursions from the variance parameters 
in OAMP, we use the notation $\bar{v}_{\B\to\A,t}[\ell]$, 
$\bar{\eta}_{\A,t}[\ell]$, $\bar{v}_{\A\to\B,t}[\ell]$, 
$\bar{v}_{\A\to \B,t}^{\mathrm{suf}}[l]$, $\bar{v}_{\B,t+1}[l]$, 
$\bar{\eta}_{\B,t}[\ell]$, and $\bar{v}_{\B\to\A,t+1}[\ell]$ instead of 
those without the bars. They are recursively defined in the same manner 
as in lines~\ref{initialization}, \ref{LMMSE}, \ref{eta_A}, \ref{var_AB}, 
\ref{var_suf}, \ref{var_B}, \ref{eta_B}, and \ref{var_BA}. The exceptional 
steps are in lines~\ref{eta_A} and \ref{var_B}, which are replaced with 
\begin{equation}
\bar{\eta}_{\A,t}[\ell] 
= \lim_{M=\delta N\to\infty}\frac{1}{N_{\mathrm{c}}[\ell]}\mathrm{Tr}
\left(
 \boldsymbol{I}_{N_{\mathrm{c}}[\ell]} 
 - \boldsymbol{W}_{t}^{\mathrm{T}}[\ell]\boldsymbol{A}[\ell]
\right) 
\end{equation}
in the large system limit and 
\begin{equation} \label{MMSE}
\bar{v}_{\B,t+1}[l]
=\mathbb{E}\left[
 \{x_{1}[0]-f_{\mathrm{opt}}(x_{1}[0]+\sqrt{v}z;v)\}^{2}
\right]
\end{equation} 
with $v=\bar{v}_{\A\to \B,t}^{\mathrm{suf}}[l]$, respectively. The limit 
in $\bar{\eta}_{\A,t}[\ell]$ can be evaluated in closed form when the 
asymptotic eigenvalue distribution of $\boldsymbol{A}^{\mathrm{T}}[\ell]
\boldsymbol{A}[\ell]$ is available. Note that the obtained 
state evolution recursions are deterministic. 

\begin{theorem} \label{theorem_OAMP}
Consider no damping $\zeta=1$ and postulate Assumptions~\ref{assumption_x} 
and \ref{assumption_A}. Furthermore, suppose that the Bayes-optimal denoiser 
$f_{\mathrm{opt}}$ is Lipschitz-continuous and nonlinear. 
\begin{itemize}
\item The empirical error covariance $N^{-1}(\boldsymbol{x}_{\B,t'+1}[l]
-\boldsymbol{x}[l])(\boldsymbol{x}_{\B,t+1}[l]
-\boldsymbol{x}[l])$ for OAMP converges almost surely to 
$\bar{v}_{\B,t+1}[l]$ given in (\ref{MMSE}) for 
all $t'\in\{0,\ldots,t\}$ in the large system limit. 
\item The state evolution recursions for Bayes-optimal OAMP converge to a 
fixed point as $t\to\infty$. 
\end{itemize}
\end{theorem}
\begin{proof}
The assumptions on the Bayes-optimal denoiser, as well as 
Assumption~\ref{assumption_A}, are used to prove 
Assumption~\ref{assumption_function}. 
The former statement is proved by confirming the inclusion of the error model 
for OAMP into the general error model~(\ref{b})--(\ref{q_ext}). The proof of 
the latter statement is based on the LM-MP strategy in \cite{Takeuchi221}: 
LM-OAMP is constructed and evaluated via Theorem~\ref{theorem_SE}. 
Furthermore, the obtained state evolution recursions for LM-OAMP are proved 
to converge and to be equivalent to those for OAMP without memory.  
See \cite[Theorem~4]{Takeuchi222} for the details. 
\end{proof}

Theorem~\ref{theorem_OAMP} justifies the definitions of the variance 
parameters in OAMP. The asymptotic performance of OAMP can be evaluated 
by solving the state evolution resursions for OAMP. 

\section{Numerical Results}
OAMP for the spatially coupled system~(\ref{spatial_coupling}) is numerically 
compared to that for the conventional system without spatial coupling. 
We used the uniform coupling coefficient $\gamma[\ell][\ell-w]=(W+1)^{-1/2}$ 
in (\ref{spatial_coupling}). In particular, $W=0$ implies no spatial coupling. 

The elements of the signal vector $\boldsymbol{x}[l]$ were sampled 
from $\mathcal{N}(0, 1/\rho)$ with probability~$\rho\in[0, 1]$ uniformly and 
randomly. Otherwise, they took zero with probability $1-\rho$. 

The sensing matrices $\{\boldsymbol{A}[\ell][l]\}$ or equivalently 
$\boldsymbol{A}[\ell]$ in (\ref{vector_system}) was postulated to have the 
SVD $\boldsymbol{A}[\ell]=\boldsymbol{\Sigma}[\ell]
\boldsymbol{V}^{\mathrm{T}}[\ell]$. The singular values in 
$\boldsymbol{\Sigma}[\ell]$ are uniquely determined from condition number 
$\kappa\geq1$ and power normalization. See \cite[Corollary 6]{Takeuchi222} 
for the details. The $|\mathcal{W}[\ell]|N\times |\mathcal{W}[\ell]|N$ 
orthogonal matrix $\boldsymbol{V}[\ell]$ is the Hadamard matrix with random 
column permutation, which is a low-complexity alternative of Haar-distributed 
orthogonal matrices in Assumption~\ref{assumption_A}.  

\begin{figure}[t]
\begin{center}
\includegraphics[width=\hsize]{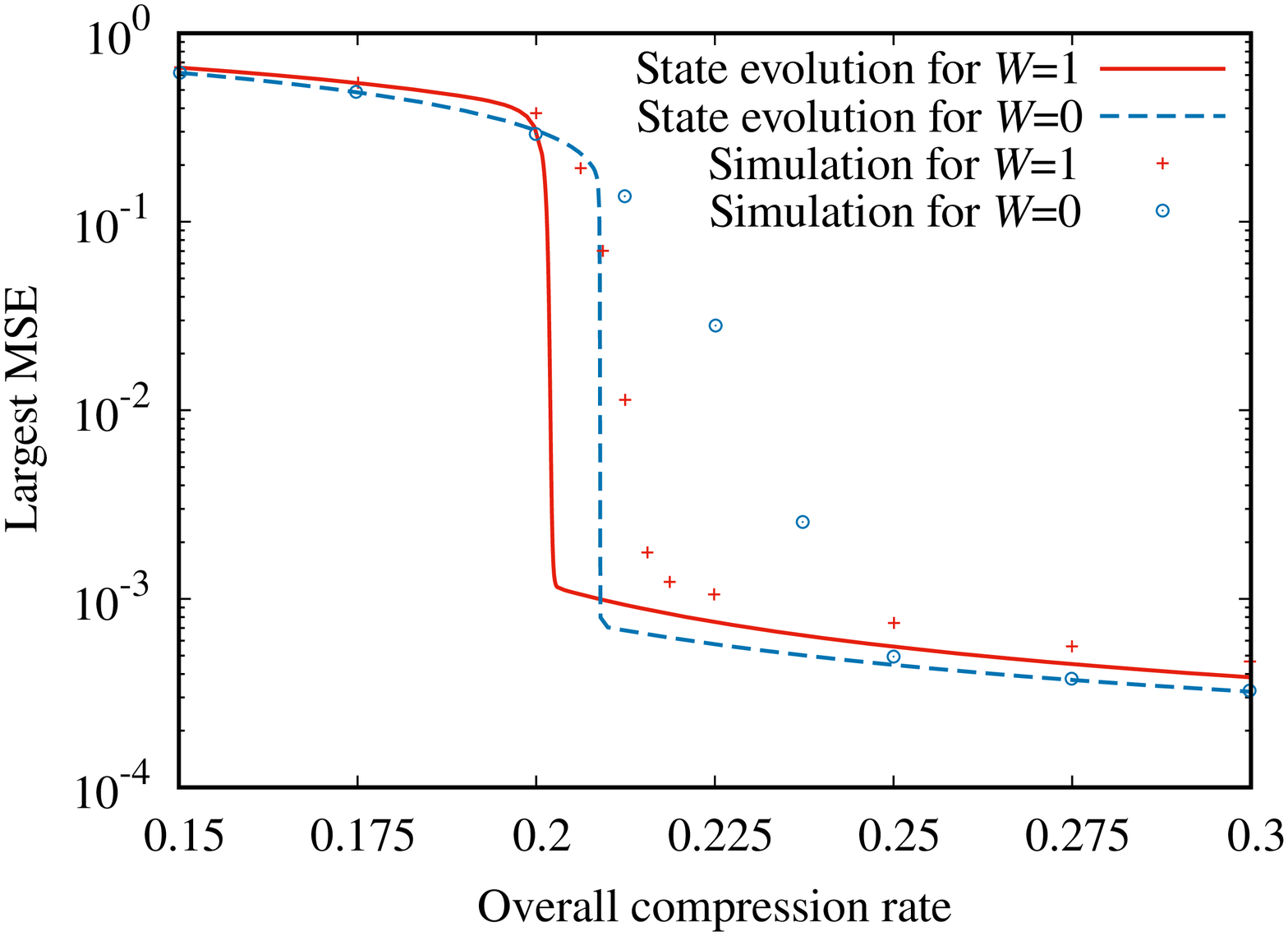}
\caption{
Largest MSE versus the overall compression rate $(1+W/L)\delta$ for 
OAMP, sensing matrices with $\kappa=10$, 
$(L,W)=(16, 1), (1,0)$, 
$N=2^{12}$, $\rho=0.1$, $1/\sigma^{2}=30$~dB, and $200$~iterations. 
}
\label{fig1} 
\end{center}
\end{figure}

Figure~\ref{fig1} shows the largest mean-square error (MSE) of OAMP over 
all sections. $10^4$ independent trials were simulated for the spatially 
coupled system with $(L,W)=(16,1)$ while $10^5$ independent trials were 
simulated for $(L,W)=(1,0)$. 
The damping factor $\zeta$ in lines~\ref{mean_damp} and 
\ref{var_damp} of Algorithm~\ref{alg1} was optimized for each $\delta=M/N$ 
via exhaustive search. 

OAMP for the spatial coupling case $W=1$ is superior to for $W=0$ in the 
so-called waterfall region, where the MSE decreases rapidly as the compression 
rate increases slightly. This result is consistent with existing results 
on spatial coupling~\cite{Kudekar11,Krzakala12,Donoho13,Takeuchi15}. 
On the other hand, the MSE in the spatial coupling case degrades slightly 
for large $\delta$. The latter result is a peculiar phenomenon for the 
spatially coupled system with orthogonally invariant sensing matrices. 

This phenomenon results from two reasons: A minor reason is due to the loss 
in the overall compression rate $(1+W/L)\delta$. This influence vanishes in 
the limit $L\to\infty$, as observed in 
\cite{Kudekar11,Krzakala12,Donoho13,Takeuchi15}. The other major reason is 
in the $|\mathcal{W}[\ell]|$-dependencies of the empirical eigenvalue 
distribution of $\boldsymbol{A}^{\mathrm{T}}[\ell]\boldsymbol{A}[\ell]
\in\mathbb{R}^{|\mathcal{W}[\ell]|N\times|\mathcal{W}[\ell]|N}$ in (\ref{vector_system}). 
To reduce the latter influence---remains even in the limit 
$L\to\infty$---small coupling width $W$ should be used.

\bibliographystyle{IEEEtran}
\bibliography{IEEEabrv,kt-icassp2023}

\end{document}